\newcommand{\sv}[1]{}
\documentclass{llncs}
 
\usepackage{times}

\let\phi=\varphi
\let\uplus=\oplus

\newcommand{\cP}{\mathcal{P}}

\newcommand{\cF}{\mathcal{F}}
\newcommand{\cH}{\mathcal{H}}

\newcommand{\No}{{\sc No}}

\newcommand{\algo}{{\sf Detectbd}}
\newcommand{\og}{{\sf Type}}

\usepackage{tikz}
\usetikzlibrary{positioning,calc}
 
\usepackage{thm-restate}
\usepackage{url,verbatim}
\usepackage{amsmath,amssymb}
\usepackage{xspace,tikz,enumitem}
\usepackage{mathrsfs,verbatim}
\usepackage{enumitem,todonotes,microtype}
\usepackage{booktabs}
\usepackage[hyperfootnotes=false]{hyperref}
\usepackage{color}
\usepackage{boxedminipage}

\usepackage[]{algorithm2e}

\newcommand{\bigoh}{\mathcal{O}}

\newcommand{\SB}{\{\,} \newcommand{\SM}{\;{|}\;} \newcommand{\SE}{\,\}}

\newcommand{\var}{\text{\normalfont{\bfseries var}}}

\newcommand{\BB}{\mathcal{B}} 
\newcommand{\CC}{\mathcal{C}}

\def\AA{{\mathcal A}}

\def\PP{{\mathcal P}}

\newcommand{\W}[1]{\text{\normalfont W}[#1]}

\newcommand{\VCSP}{\textsc{VCSP}}
\newcommand{\CSP}{\textsc{CSP}}

\newcommand{\ol}[1]{\overline{#1}}

\def\MS#1{\mbox{MSO}}

\spnewtheorem{ourfact}{Fact}{\bfseries}{\itshape}
\spnewtheorem{observation}{Observation}{\bfseries}{\itshape}
\newcommand\blfootnote[1]{%
  \begingroup
  \renewcommand\thefootnote{}\footnote{#1}%
  \addtocounter{footnote}{-1}%
  \endgroup
}

\let\doendproof\endproof
\renewcommand\endproof{~\hfill$\qed$\doendproof}

\begin{document}

\title{Backdoors to Tractable Valued CSP}

\author{Robert Ganian \and M. S. Ramanujan \and Stefan Szeider}

\institute{Algorithms and Complexity Group, TU Wien, Vienna, Austria}

\maketitle

\begin{abstract}
  We extend the notion of a strong backdoor from the CSP setting to
  the Valued CSP setting (VCSP, for short).  This provides a means for
  augmenting a class of tractable VCSP instances to instances that are
  outside the class but of small distance to the class, where the
  distance is measured in terms of the size of a smallest backdoor.
  We establish that VCSP is fixed-parameter tractable when
  parameterized by the size of a smallest backdoor into every tractable
  class of VCSP instances characterized by a (possibly infinite)
  tractable valued constraint language of finite arity and finite
  domain. We further extend this fixed-parameter tractability result
  to so-called ``scattered classes'' of VCSP instances where each
  connected component may belong to a different tractable class.
\end{abstract}


\section{Introduction}
\label{sec:intro}

Valued CSP (or VCSP for short) is a powerful framework that entails among others
\blfootnote{The authors acknowledge support by the Austrian Science Fund (FWF, project P26696). Robert Ganian is also affiliated with FI MU, Brno, Czech Republic.}the problems CSP and MAX-CSP as special cases \cite{Zivny12}.  A VCSP
instance consists of a finite set of cost functions over a finite set
of variables which range over a domain $D$, and the task is to
find an instantiation of these variables that minimizes the sum of the
cost functions. The VCSP framework is robust and has  been studied
in different contexts in computer science.  In its full generality,
VCSP considers cost functions that can take as values the rational
numbers and positive infinity.  CSP (feasibility) and Max-CSP
(optimisation) arise as special cases by limiting the values of cost
functions to $\{0,\infty\}$ and $\{0,1\}$, respectively. Clearly VCSP
is in general intractable. Over the last decades much research has
been devoted into the identification of tractable VCSP subproblems. An
important line of this research (see, e.g.,
\cite{JeavonsKrokhinZivny14,KolmogorovZivny,ThapperZivny15}) is the
characterization of tractable VCSPs in terms of restrictions on
the underlying \emph{valued constraint language} $\Gamma$, i.e., a set~$\Gamma$ of cost functions that guarantees polynomial-time solvability
of all VCSP instances that use only cost functions from $\Gamma$.
The VCSP restricted to instances with cost functions from $\Gamma$ is
denoted by $\textsc{VCSP}[\Gamma]$. 
 
In this paper we provide algorithmic results which allow us to
gradually augment a tractable VCSP based on the notion of a (strong)
\emph{backdoor} into a tractable class of instances, called the
\emph{base class}. Backdoors where introduced by Williams \emph{et al.}~\cite{WilliamsGomesSelman03,WilliamsGomesSelman03a} for SAT and
CSP and generalize in a natural way to VCSP{}. Let $\mathcal{C}$
denote a tractable class of VCSP instances over a finite domain $D$. A
backdoor of a VCSP instance $\cP$ into~$\mathcal{C}$ is a (small)
subset $B$ of the variables of $\cP$ such that for all partial
assignments $\alpha$ that instantiate $B$, the restricted instance
$\cP|_\alpha$ belongs to the tractable class $\mathcal{C}$. Once we
know such a backdoor~$B$ of size $k$ we can solve $\cP$ by solving at
most $|D|^k$ tractable instances. In other words, VCSP is then
\emph{fixed parameter tractable} parameterized by backdoor size. This
is highly desirable as it allows us to scale the
tractability for $\mathcal{C}$ to instances
outside the class, paying for an increased ``distance'' from
$\mathcal{C}$ only by a larger constant factor.

In order to apply this backdoor approach to solving a \VCSP{}
instance, we first need to \emph{find} a small backdoor.
This turns out to be an algorithmically challenging task. The
fixed-parameter tractability of backdoor detection has been subject
of intensive research in the context of SAT (see, e.g.,
\cite{GaspersSzeider12}) and CSP (see, e.g.,
\cite{CarbonnelCooper16}). In this paper we extend this line of
research to \VCSP.

First we  obtain some basic and fundamental results on backdoor
detection when the base class is defined by a valued constraint
language $\Gamma$. We obtain fixed-parameter tractability 
for the detection of backdoors into $\VCSP[\Gamma]$ where
$\Gamma$ is a valued constraint language with cost functions of  bounded arity. In fact, we show the stronger result: fixed-parameter
tractability also holds with respect to \emph{heterogeneous} base
classes of the form
$\VCSP[\Gamma_1] \cup \dots \cup \VCSP[\Gamma_\ell]$ where different
assignments to the backdoor variables may result in instances that
belong to different base classes $\VCSP[\Gamma_i]$.  A similar result
holds for \CSP{}, but the \VCSP{} setting is slightly more complicated as a
valued constraint language of finite arity over a finite domain is not
necessarily finite.

Secondly, we extend the basic fixed-parameter tractability result to
so-called \emph{scattered} base classes of the form
$\VCSP[\Gamma_1]\uplus \dots \uplus \VCSP[\Gamma_\ell]$ which contain
\VCSP{} instances where each connected component belongs to a tractable
class $\VCSP[\Gamma_i]$ for some $1\leq i \leq \ell$---again in the
heterogeneous sense that for different assignments to the backdoor
variables a single component of the reduced instance may belong to
different classes $\VCSP[\Gamma_i]$. Backdoors into a scattered base
class can be much smaller than backdoors into each single class it is
composed of, hence the gain is huge if we can handle scattered
classes. This boost in scalability does not come for free. Indeed,
already the ``crisp'' case of CSP, which was the topic of a recent
SODA paper \cite{GanianRamanujanSzeider16}, requires a sophisticated
algorithm which makes use of advanced techniques from
parameterized algorithm design. This algorithm works under the
requirement that the constraint languages contain all unary
constraints (i.e., is conservative); this is a reasonable requirement
as one needs these unary cost functions to express partial assignments
(see also Section~\ref{sec:prelim} for further discussion). Here we lift the crisp case to general \VCSP, and this also represents our main technical contribution.

To achieve this, we proceed in two phases. First we transform the backdoor detection problem from a general scattered class
$\VCSP(\Gamma_1)\uplus \dots \uplus \VCSP(\Gamma_\ell)$ to a scattered
class $\VCSP(\Gamma_1')\uplus \dots \uplus \VCSP(\Gamma_\ell')$ over
\emph{finite} valued constraint languages~$\Gamma_i'$. In the subsequent second
phase we transform the problem to a backdoor detection problem into a
scattered class
$\VCSP(\Gamma_1'')\uplus \dots \uplus \VCSP(\Gamma_\ell'')$ where each
$\Gamma_i''$ is a finite crisp language; i.e., we reduce from the \VCSP{}
setting to the \CSP{} setting. We believe that this sheds light on an interesting link between backdoors in the \VCSP{} and \CSP{} settings.
The latter problem can now be solved using the
known algorithm~\cite{GanianRamanujanSzeider16}.

\subsection*{Related Work}
Williams \emph{et al.}~\cite{WilliamsGomesSelman03,WilliamsGomesSelman03a}
introduced backdoors for CSP or SAT as a theoretical tool to capture
the overall combinatorics of instances. The purpose was an analysis of
the empirical behaviour of backtrack search algorithms. Nishimura
et. al~\cite{NishimuraRagdeSzeider04-informal} started the
investigation on the parameterized complexity of \emph{finding} a
small SAT backdoor and using it to solve the instance. This lead to a
number of follow-up work (see \cite{GaspersSzeider12}). Parameterized
complexity provides here an appealing framework, as given a CSP
instance with $n$ variables, one can trivially find a backdoor of size
$\leq k$ into a fixed tractable class of instances by trying all
subsets of the variable set containing $\leq k$ variables; but there are $\Theta(n^k)$
such sets, and therefore the running time of this brute-force
algorithm scales very poorly in~$k$. Fixed-parameter tractability
removes $k$ from the exponent providing running times of the form
$f(k)n^{c}$ which yields a significantly better scalability in backdoor
size.

Extensions to the basic notion of a
backdoor have been proposed, including backdoors with empty clause
detection \cite{DilkinaGomesSabharwal07}, backdoors in the context of
learning \cite{DilkinaGomesSabharwal09}, heterogeneous backdoors where
different instantiations of the backdoor variables may result in
instances that belong to different base
classes~\cite{GaspersMisraOSZ14}, and backdoors into scattered classes
where each connected component of an instance may belong to a
different tractable class~\cite{GanianRamanujanSzeider16}.  Le Bras \emph{et
al.}~\cite{LenbrasBernsteinGomesSelmanDover13} used backdoors to
speed-up the solution of hard problems in materials discovery, using a
crowd sourcing approach to find small backdoors.

The research on the parameterized complexity of backdoor detection was
also successfully extended to other problem areas including
disjunctive answer set programming~\cite{FichteSzeider15,FichteSzeider15b}, abstract
argumentation~\cite{DvorakOrdyniakSzeider12}, and  integer
linear programming \cite{GanianOrdyniak16}.
There are also several papers that investigate the parameterized
complexity of backdoor detection for CSP{}.  Bessi{\`e}re \emph{et
al.}~\cite{BessiereCarbonnelHebrardKatsirelosWalsh13}, considered
``partition backdoors'' which are sets of variables whose deletion
partitions the CSP instance into two parts, one falls into a tractable
class defined by a conservative polymorphism, and the other part is a
collection of independent constraints. They also performed an
empirical evaluation of the backdoor approach which resulted in
promising results.  Gaspers \emph{et al.}~\cite{GaspersMisraOSZ14} considered
heterogeneous backdoors into tractable CSP classes that are
characterized by  polymorphisms. A similar
approach was also undertaken by Carbonnel \emph{et
al.}~\cite{CarbonnelCooperHebrard14} who also considered base classes
that are ``$h$-Helly'' for a fixed integer $h$ under the additional
assumption that the domain is a finite subset of the natural numbers
and comes with a fixed ordering.

\section{Preliminaries}
\label{sec:prelim}
\subsection{Valued Constraint Satisfaction}
For a tuple $t$, we shall denote by $t[i]$ its $i$-th component.
We shall denote by $\cal Q$ the set of all rational numbers, by ${\cal
  Q}_\geq 0$ the set of all nonnegative rational numbers, and by
${\cal \ol Q}_{\geq 0}$ the set of all nonnegative rational numbers
together with positive infinity, $\infty$. We define
$\alpha+\infty=\infty+\alpha=\infty$ for all $\alpha\in {\cal \ol
  Q}_{\geq 0}$, and $\alpha \cdot \infty = \infty$ for all $\alpha\in
{\cal Q}_{\geq 0}$. The elements of ${\cal \ol Q}_{\geq 0}$ are called \emph{costs}.

For every fixed set $D$ and $m\geq 0$, a function $\phi$ from $D^m$ to ${\cal \ol Q}_{\geq 0}$ will be called a \emph{cost function} on $D$ of arity $m$. $D$ is called the \emph{domain}, and here we will only deal with finite domains.
If the range of $\phi$ is $\{0,\infty\}$, then $\phi$ is called a \emph{crisp} cost function. 

With every relation $R$ on $D$, we can associate a crisp cost function $\phi_R$ on $D$ which maps tuples in $R$ to $0$ and tuples not in $R$ to $\infty$. On the other hand, with every $m$-ary cost function $\phi$ we can associate a relation $R_\phi$ defined by $(x_1,\dots,x_m)\in R_\phi \Leftrightarrow \phi(x_1,\dots,x_m)< \infty$. In the view of the close correspondence between crisp cost functions and relations we shall use these terms interchangeably in the rest of the paper.

A \textsc{VCSP} instance consists of a set of variables, a set of possible values, and a multiset of valued constraints. Each valued constraint has an associated cost function which assigns a cost to every possible tuple of values for the variables in the scope of the valued constraint. The goal is to find an assignment of values to all of the variables that has the minimum total cost. A formal definition is provided below.

\begin{definition}[\textsc{VCSP}]
An instance $\cal P$ of the \textsc{Valued Constraint satisfaction Problem}, or \textsc{VCSP}, is a triple $(V,D,\cal C)$ where $V$ is a finite set of \emph{variables}, which are to be assigned values from the set $D$, and $\cal C$ is a multiset of \emph{valued constraints}. Each $c\in \cal C$ is a pair $c=(\vec{x},\phi)$, where $\vec{x}$ is a tuple of variables of length $m$ called the \emph{scope} of $c$, and $\phi: D^m\rightarrow {\cal \ol Q}_{\geq 0}$ is an $m$-ary cost function. An \emph{assignment} for the instance $\cal P$ is a mapping $\tau$ from $V$ to $D$. We extend $\tau$ to a mapping from $V^k$ to $D^k$ on tuples of variables by applying $\tau$ componentwise. The \emph{cost} of an assignment $\tau$ is defined as follows:

\[\text{Cost}_{\cal P}(\tau)=\sum_{(\vec{x},\phi)\in \cal C}\phi(\tau(\vec{x})).\]

The task for VCSP is the computation of an assignment with minimum
cost, called a \emph{solution} to $\cal P$.
\end{definition}

For a constraint $c$, we will use $\var(c)$ to denote the set of variables which occur in the scope of $c$. We will later also deal with the \emph{constraint satisfaction problem}, or \CSP{}. Having already defined {\VCSP}, it is advantageous to simply define {\CSP} as the special case of {\VCSP} where each valued constraint has a crisp cost function.

The following representation of a cost function will sometimes be useful for our purposes. A \emph{cost table} for an $m$-ary cost function $\phi$ is a table with $D^m$ rows and $m+1$ columns with the following property: each row corresponds to a unique tuple $\vec{a}=(a_1,\dots,a_m)\in D^m$, for each $i\in [m]$ the position $i$ of this row contains $a_i$, and position $m+1$ of this row contains $\phi(a_1,\dots,a_m)$.

A \emph{partial assignment} is a mapping from $V'\subseteq V$ to $D$. Given a partial assignment $\tau$, the \emph{application} of $\tau$ on a valued constraint $c=(\vec{x},\phi)$ results in a new valued constraint $c|_{\tau}=(\vec{x}',\phi')$ defined as follows. Let $\vec{x}'=\vec{x}\setminus V'$ (i.e., $\vec{x}'$ is obtained by removing all elements in $V\cap \vec{x}$ from $\vec{x}$) and $m'=|\vec{x'}|$. Then for each $\vec{a'}\in D^{m'}$, we set $\phi'(\vec{a'})=\phi(\vec{a})$ where for each $i\in [m]$
\[\vec{a}[i]= 
 \left\{ \begin{array}{ll}
                              \tau(\vec{x}[i])&  \mbox{\ if } \vec{x}[i]\in V'  \\
                             \vec{a'}[i-j] & \mbox{\ otherwise, where } j=|\SB \vec{x}[p] \SM p\in [i] \SE \cap V'|.                             
                            \end{array}
                    \right. \]
Intuitively, the tuple $\vec{a}$ defined above is obtained by taking the original tuple $\vec{a}'$ and enriching it by the values of the assignment $\tau$ applied on the ``missing'' variables from $\vec{x}$. In the special case when $\vec{x}'$ is empty, the valued constraint $c|_{\tau}$ becomes a nullary constraint whose cost function $\phi'$ will effectively be a constant.
The application of $\tau$ on a \textsc{VCSP} instance $\cal P$ then
results in a new \textsc{VCSP} instance ${\cal P}|_{\tau}=(V\setminus
V', D, \cal C')$ where ${\cal C'}=\SB c|_{\tau} \SM c\in \cal
C\SE$. It will be useful to observe that applying a partial assignment
$\tau$ can be done in time linear in $|\cal P|$ (each valued constraint
can be processed independently, and the processing of each such valued
constraint consists of merely pruning the cost table).

\subsection{Valued Constraint Languages}

A \emph{valued constraint language} (or \emph{language} for short) is
a set of cost functions. The arity of a language $\Gamma$ is the
maximum arity of a cost function in $\Gamma$, or $\infty$ if $\Gamma$
contains cost functions of arbitrarily large arities.  Each language
$\Gamma$ defines a set $\textsc{VCSP}[\Gamma]$ of \textsc{VCSP}
instances which only use cost functions from $\Gamma$; formally,
$(V, D, {\cal C})\in\textsc{VCSP}[\Gamma]$ iff each
$(\vec{x},\phi)\in \cal C$ satisfies $\phi\in \Gamma$. A language is \emph{crisp} if it contains only crisp cost functions.

A language
$\Gamma$ is \emph{globally tractable} if there exists a
polynomial-time algorithm which solves $\textsc{VCSP}[\Gamma].$\footnote{The literature also defines the notion of \emph{tractability}~\cite{JeavonsKrokhinZivny14,KrokhinBulatovJeavons05}, which we do not consider here. We remark that, to the best of our knowledge, all known tractable constraint languages are also globally tractable~\cite{JeavonsKrokhinZivny14,KrokhinBulatovJeavons05}
}.
Similarly, a class $\cal H$ of \VCSP{} instances is called tractable if there exists a polynomial-time algorithm which solves $\cal H$.
For technical reasons, we will implicitly assume that every language contains all nullary cost functions (i.e., constants); it is easily seen that adding such cost functions into a language has no impact on its tractability.

There are a few other properties of languages that will be required to formally state our results.
A language $\Gamma$ is \emph{efficiently recognizable} if there
exists a polynomial-time algorithm which takes as input a cost
function $\phi$ and decides whether $\phi\in\Gamma$. We note that
every finite language is efficiently recognizable.

A language $\Gamma$ is \emph{closed under partial assignments} if for
every instance $\cP\in \VCSP[\Gamma]$ and every partial assignment
$\tau$ on $\cP$ and every valued constraint $c=(\vec{x},\phi)$ in
$\cP$, the valued constraint $c|_{\tau}=(\vec{x}',\phi')$ satisfies
$\phi'\in \Gamma$. The \emph{closure of a language~$\Gamma$ under
  partial assignments}, is the language $\Gamma'\supseteq \Gamma$
containing all cost functions that can be obtained from $\Gamma$ via
partial assignments; 
formally, $\Gamma'$ contains a cost function
$\phi'$ if and only if there exists a cost function $\phi\in \Gamma$
such that for a constraint $c=(\vec{x},\phi)$ and an assignment
$\tau:X\rightarrow D$ defined on a subset $X\subseteq \var(c)$ we have 
$c|_{\tau}=(\vec{x}',\phi')$. 

If a language $\Gamma$ is closed under partial assignments, then also
$\VCSP[\Gamma]$ is closed under partial assignments, which is a
natural property and provides a certain robustness of the class. This
robustness is also useful when considering backdoors into
$\VCSP[\Gamma]$ (see Section~\ref{sec:bdstandard}), as then every
superset of a backdoor remains a backdoor. Incidentally, being closed
under partial assignments is also a property of tractable classes
defined in terms of a polynomial-time
subsolver~\cite{WilliamsGomesSelman03,WilliamsGomesSelman03a} where
the property is called \emph{self-reducibility}.

A language is \emph{conservative} if it contains all unary cost functions~\cite{KolmogorovZivny}.
We note that being closed under partial assignments
is  closely related to the well-studied property of
conservativeness.  Crucially, for every
conservative globally tractable language $\Gamma$, its closure under partial
assignments $\Gamma'$ will also be globally tractable; indeed, one can observe
that every instance $\cP\in \VCSP[\Gamma']$ can be converted, in linear
time, to a solution-equivalent instance $\cP'\in \VCSP[\Gamma]$ by
using infinity-valued (or even sufficiently high-valued) unary cost
functions to model the effects of partial assignments.

\subsection{Parameterized Complexity}
\label{sub:parcomp}
We give a brief and rather informal review of the most important
concepts of parameterized complexity. For an in-depth treatment of the
subject we refer the reader to other sources
\cite{CyganFKLMPPS15,DowneyFellows13,FlumGrohe06,Niedermeier06}.

The instances of a parameterized problem can be considered as pairs
$(I,k)$ where~$I$ is the \emph{main part} of the instance and $k$ is
the \emph{parameter} of the instance; the latter is usually a
non-negative integer.  A parameterized problem is
\emph{fixed-parameter tractable} (FPT) if instances $(I,k)$ of size
$n$ (with respect to some reasonable encoding) can be solved in time
$\bigoh(f(k)n^c)$ where $f$ is a computable function and $c$ is a
constant independent of $k$. The function $f$ is called the
\emph{parameter dependence}, and algorithms with running time in this
form are called \emph{fixed-parameter algorithms}. Since the parameter
dependence is usually superpolynomial, we will often give the running
times of our algorithms in $\bigoh^*$ notation which suppresses
polynomial factors.  Hence the running time of an FPT algorithm can be
simply stated as $\bigoh^*(f(k))$.

The exists a completeness theory which allows to obtain strong
theoretical evidence that a parameterized problem is \emph{not}
fixed-parameter tractable. This theory is based on a hierarchy of
parameterized complexity classes
$\W{1}\subseteq \W{2} \subseteq \dots$ where all inclusions are
believed to be proper. If a parameterized problem is shown to be
$\W{i}$-hard for some $i\geq 1$, then the problem is unlikely to be
fixed-parameter tractable, similarly to an NP-complete problem being
solvable in polynomial time
\cite{CyganFKLMPPS15,DowneyFellows13,FlumGrohe06,Niedermeier06}.

 \section{Backdoors into Tractable Languages} 
 \label{sec:bdstandard}

 This section is devoted to establishing the first general results for
 finding and exploiting backdoors for \VCSP.  We first present the
 formal definition of backdoors in the context of VCSP and describe
 how such backdoors once found, can be used to solve the VCSP
 instance.  Subsequently, we show how to detect backdoors into a
 single tractable {\VCSP} class with certain properties. In fact, our
 proof shows something stronger. That is, we show how to detect
 \emph{heterogeneous} backdoors into a finite set of VCSP classes
 which satisfy these properties. The notion of heterogeneous backdoors
 is based on that introduced by Gaspers 
 \emph{et al.}~\cite{GaspersMisraOSZ14}.
For now, we proceed with the definition of a backdoor.

 \begin{definition}\label{def:backdoors}
Let $\cal H$ be a fixed class of VCSP instances over a domain $D$
and let ${\cal P}=(V,D,\cal C)$ be a VCSP instance. A
\emph{backdoor} into $\cal H$ is a  subset $X\subseteq V$ such that
for each assignment $\tau:X\rightarrow D$, the reduced instance ${\cal
  P}|_{\tau}$ is in $\cal H$. 
 \end{definition}

 We note that this naturally corresponds to the notion of a
 \emph{strong} backdoor in the context of Constraint Satisfaction and
 Satisfiability~\cite{WilliamsGomesSelman03,WilliamsGomesSelman03a};
 here we drop the adjective ``strong'' because the other kind of
 backdoors studied on these structures (so-called \emph{weak} backdoors)
 do not seem to be useful in the general {\VCSP} setting. Namely,
 in analogy to the CSP setting, one would define a weak backdoor of a
 VCSP instance ${\cal P}=(V,D,\cal C)$ into $\cal H$ as a subset
 $X\subseteq V$ such that for some assignment $\tau:X\rightarrow D$
 (i)~the reduced instance ${\cal P}|_{\tau}$ is in $\cal H$ and
 (ii)~$\tau$ can be extended to an assignment to $V$ of minimum cost.
 However, in order to ensure (ii) we need to compare the cost of $\tau$ with  the costs of all
 other assignments $\tau'$ to~$V$. If $X$ is not a strong backdoor, then
 some of the reduced instances ${\cal P}|_{\text{$\tau'$ restricted
     to $X$}}$ will be outside of
 ${\cal H}$, and so in general we have no efficient way of
 determining a minimum cost assignment for it.

We begin by showing that small backdoors for globally tractable languages can always be used to efficiently solve {\VCSP} instances as long as the domain is finite (assuming such a backdoor is known).

\begin{lemma}
\label{lem:using}
Let $\cal H$ be a tractable class of \VCSP{} instances over a finite domain $D$. There exists an algorithm which takes as input a $\VCSP$ instance $\cal P$ along with a backdoor $X$ of $\PP=(V,D,\cal C)$ into $\cal H$, runs in time $\bigoh^*(|D|^{|X|})$, and solves $\PP$.
\end{lemma}

\begin{proof}
Let $\BB$ be a polynomial-time algorithm which solves every $\PP$ in $\cal H$, i.e., outputs a minimum-cost assignment in $\PP$; the existence of $\BB$ follows by the tractability of $\cal H$. Consider the following algorithm $\AA$. First, $\AA$ branches on the at most $|D|^{|X|}$-many partial assignments of $X$. In each branch, $\AA$ then applies the selected partial assignment~$\tau$ to obtain the instance $\PP|_{\tau}$ in linear time. In this branch, $\AA$ proceeds by calling $\BB$ on $\PP|_{\tau}$, and stores the produced assignment along with its cost. After the branching is complete $\AA$ reads the table of all of the at most $|D|^{|X|}$ assignments and costs outputted by $\BB$, and selects one assignment (say $\alpha$) with a minimum value (cost) $a$. Let $\tau$ be the particular partial assignment on $X$ which resulted in the branch leading to $\alpha$. $\AA$ then outputs the assignment $\alpha \cup \tau$ along with the value (cost) $a$.
\end{proof}

Already for crisp languages it is known that having a small backdoor
does not necessarily allow for efficient (i.e., fixed-parameter) algorithms when the domain is not bounded. Specifically, the \W{1}-hard $k$-clique problem can be encoded into a \CSP{} with only~$k$ variables~\cite{PapadimitriouYannakakis99}, which naturally contains a backdoor of size at most $k$ for every crisp language under the natural assumption that the language contains the empty constraint. Hence the finiteness of the domain in Lemma~\ref{lem:using} is a necessary condition for the statement to hold.
   
Next, we show that it is possible to find a small backdoor into 
 $\VCSP[\Gamma]$ efficiently (or correctly determine that no such small backdoor exists) as long as $\Gamma$ has two properties. First, $\Gamma$ must be efficiently recognizable; it is easily seen that this condition is a necessary one, since detection of an empty backdoor is equivalent to determining whether the instance lies in $\VCSP[\Gamma]$. Second, the arity of $\Gamma$ must be bounded. This condition is also necessary since already in the more restricted \CSP{} setting it was shown that backdoor detection for a wide range of natural crisp languages (of unbounded arity) is \W{2}-hard~\cite{GaspersMisraOSZ14}. 

Before we proceed, we introduce the notion of heterogeneous backdoors for \VCSP{} which represent a generalization of backdoors into classes defined in terms of a single language.
For languages $\Gamma_1,\dots,\Gamma_\ell$, a heterogeneous backdoor
is a backdoor into the class ${\cal H}=\VCSP[\Gamma_1]\cup \dots \cup
\VCSP[\Gamma_\ell]$; in other words, after each assignment to the
backdoor variables, all cost functions in the resulting instance must belong to a language from our set.
We now show that detecting small heterogenous backdoors is fixed-parameter tractable parameterized by the size of the backdoor.

\enlargethispage*{5mm}

\begin{lemma}
\label{lem:finding}
Let $\Gamma_1,\dots, \Gamma_\ell$ be efficiently recognizable languages over a domain $D$ of size at most $d$ and let $q$ be a bound on the arity of $\Gamma_i$ for every $i\in [\ell]$.
There exists an algorithm which takes as input a $\VCSP$ instance $\cal P$ over $D$ and an integer $k$, runs in time $O^*( (\ell \cdot d \cdot  (q+1))^{k})$, and either outputs a backdoor $X$ of $\cal P$ into $\VCSP[\Gamma_1]\cup \dots \cup \VCSP[\Gamma_\ell]$ such that $|X|\leq k$ or correctly concludes that no such backdoor exists.
\end{lemma}

\begin{proof}
  The algorithm is a standard branching algorithm (see also
  \cite{GaspersMisraOSZ14}). Formally, the algorithm is called 
  $\algo$, takes as input an instance $\cP=(V,D,\cal C)$, integer $k$, a set of variables $B$ of size at most $k$ and  in time
  $O^*((\ell\cdot d \cdot (q+1))^k)$ either correctly concludes that $\cal P$ has
  no  backdoor $Z\supseteq B$ of size at most $k$ into $\VCSP[\Gamma_1]\cup \dots \cup \VCSP[\Gamma_\ell]$ or returns a backdoor $Z$ of $\cal P$  into $\VCSP[\Gamma_1]\cup \dots \cup \VCSP[\Gamma_\ell]$ of size at most
  $k$. The algorithm is initialized with $B=\emptyset$.

In the base case, 
if $|B|=k$, and $B$ is a backdoor of $\cP$ into $\VCSP[\Gamma_1]\cup \dots \cup \VCSP[\Gamma_\ell]$ then we return the set $B$. Otherwise, we return {\No}. We now move to the description of the case when $|B|<k$.

In this case, if for every $\sigma:B\to D$ there is an $i\in [\ell]$ such that $\cP|_{\sigma}\in \VCSP[\Gamma_i]$, then it sets $
  Z=B$ and returns it. That is, if $B$ is already found to be a backdoor of the required kind, then the algorithm returns $B$. 
  Otherwise, it computes an assignment $\sigma:B\to D$ and valued constraints $c_1,\dots, c_\ell$ in $\cP|_{\sigma}$  such that for every $i\in [\ell]$, the cost function of $c_i$ is not in 
 $\Gamma_i$. 
  Observe that for some $\sigma$, such a set of constraints must exist. Furthermore, since every $\Gamma_i$ is efficiently recognizable and $B$ has size at most $k$, the selection of these valued constraints takes time $\bigoh^*(d^k)$. 
The algorithm now constructs a set $Y$ as follows. Initially, $Y=\emptyset$. For each $i\in [\ell]$, if the scope of the
constraint $c_i$ contains more than $q$ variables then it adds to  $Y$
an arbitrary $q+1$-sized subset of the scope of $c_i$. Otherwise, it adds to $Y$ all the variables in the scope of $c$. This completes the definition of $Y$. Observe that any backdoor set for the given instance which contains $B$ must also intersect $Y$. Hence the algorithm now branches on the set $Y$. Formally, for every $x\in Y$ it executes the recursive calls $\algo(\cP,k,B\cup \{x\})$. If for some $x\in Y$, the invoked call returned a set of variables, then it must be a backdoor set of the given instance and hence it is returned. Otherwise, the algorithm returns {\No}.

%
%

    Since the branching factor of this algorithm is at most $\ell \cdot (q+1)$ and
   the set $B$, whose size is upper bounded by $k$, is enlarged with each recursive call, the number of nodes in the search tree is bounded by $\bigoh((\ell\cdot (q+1))^k)$. Since the time spent at each node is bounded by $\bigoh^*(d^k)$, 
    the running time of the algorithm $\algo$ is bounded by $\bigoh^*((\ell\cdot (q+1)\cdot d)^k)$. 
\end{proof}
Combining Lemmas~\ref{lem:using} and~\ref{lem:finding}, we obtain the
main result of this section.

\enlargethispage*{5mm}

\begin{corollary}
\label{cor:bd}
Let $\Gamma_1,\dots, \Gamma_\ell$ be globally tractable and efficiently recognizable languages each of  arity at most $q$ over a domain of size $d$. There exists an algorithm which solves {\VCSP} in time $O^*( (\ell \cdot d\cdot  (q+1))^{k^2+k})$, where  $k$ is the size of a minimum backdoor of the given instance into $\VCSP[\Gamma_1]\cup \dots \cup \VCSP[\Gamma_\ell]$.
\end{corollary}

 \section{Backdoors into Scattered Classes}
 \label{sec:bdscattered}
Having established Corollary~\ref{cor:bd} and knowing that both the
arity and domain restrictions of the language are necessary, it is
natural to ask whether it is possible to push the frontiers of
tractability for backdoors to more general classes of VCSP
instances. In particular, there is no natural reason why the instances
we obtain after each assignment into the backdoor should necessary
always belong to the same language $\Gamma$ even if $\Gamma$ itself is one among several globally tractable languages. In fact, it is not
difficult to show that as long as each ``connected component'' of the
instance belongs to some tractable class after each assignment into
the backdoor, then we can use the backdoor in a similar fashion as in
Lemma~\ref{lem:using}. Such a generalization of backdoors from single
languages to collections of languages has recently been obtained in
the \CSP{} setting~\cite{GanianRamanujanSzeider16} for conservative
constraint languages. We proceed by formally defining these more general classes of \VCSP{} instances, along with some other required notions.

\subsection{Scattered Classes}
\label{sub:scattered}
A {\VCSP} instance $(V,D,\cal C)$ is \emph{connected} if for each
partition of its variable set into nonempty sets $V_1$ and $V_2$, there exists
at least one constraint $c\in \cal C$ such that $\var(c)\cap V_1\neq \emptyset$
and $\var(c)\cap V_2\neq \emptyset$. A \emph{connected component} of
$(V,D,\cal C)$  is a maximal connected subinstance $(V',D,{\cal
  C}')$ for $V'\subseteq V$, ${\cal C}'\subseteq \cal C$.
These notions naturally correspond to the connectedness and connected components of standard graph  representations of {\VCSP} instances. 


Let $\Gamma_1,\dots,\Gamma_d$ be languages. Then the \emph{scattered class} $\VCSP(\Gamma_1)\uplus \dots \uplus \VCSP(\Gamma_d)$ is the class of all instances $(V,D,\cal C)$ which may be partitioned into pairwise variable disjoint subinstances $(V_1,D,{\cal C}_1),\dots,(V_d,D,{\cal C}_d)$ such that $(V_i,D,{\cal C}_i)\in \VCSP[\Gamma_i]$ for each $i\in [d]$. Equivalently, an instance $\PP$ is in $\VCSP(\Gamma_1)\uplus \dots \uplus \VCSP(\Gamma_d)$ iff each connected component in $\PP$ belongs to some $\VCSP[\Gamma_i]$, $i\in [d]$.

\begin{lemma}
Let $\Gamma_1,\dots, \Gamma_d$ be globally tractable languages. Then there exists a polynomial-time algorithm solving {\VCSP} for all instances $P\in \VCSP(\Gamma_1)\uplus \dots \uplus \VCSP(\Gamma_d)$.
\end{lemma}


It is worth noting that while scattered classes on their own are a somewhat trivial extension of the tractable classes defined in terms of individual languages, backdoors into scattered classes can be much smaller than backdoors into each individual globally tractable language (or, more precisely, each individual class defined by a globally tractable language). That is because a backdoor can not only simplify cost functions to ensure they belong to a specific language, but it can also disconnect the instance into several ``parts'', each belonging to a different language, and furthermore the specific language each ``part'' belongs to can change for different assignments into the backdoor. 
As a simple example of this behavior, consider the boolean domain, let
$\Gamma_1$ be the globally tractable crisp language corresponding to
Horn constraints~\cite{Schaefer78}, and let $\Gamma_2$ be a globally
tractable language containing only submodular cost
functions~\cite{CohenJeavonsZivny08}. It is not difficult to construct
an instance $\PP=(V_1\cup V_2\cup \{x\},\{0,1\},\cal C)$ such that 
(\textbf{a}) every assignment to $x$ disconnects $V_1$ from $V_2$,
(\textbf{b}) in $\PP|_{x\mapsto 0}$, all valued constraints over $V_1$ are crisp Horn constraints and all valued constraints over $V_2$ are submodular, and
(\textbf{c}) in $\PP|_{x\mapsto 1}$, all valued constraints over $V_1$ are submodular and all valued constraints over $V_2$ are crisp Horn constraints.
In the hypothetical example above, it is easy to verify that $x$ is a backdoor into $\VCSP[\Gamma_1]\uplus \VCSP[\Gamma_2]$ but the instance does not have a small backdoor into neither $\VCSP[\Gamma_1]$ nor $\VCSP[\Gamma_2]$.

It is known that backdoors into scattered classes can be used to obtain
fixed-parameter algorithms for \CSP{}, i.e., both finding and using
such backdoors is FPT when dealing with crisp languages of bounded
arity and domain size~\cite{GanianRamanujanSzeider16}. Crucially,
these previous results relied on the fact that every crisp language of
bounded arity and domain size is finite (which is not true for valued
constraint languages in general). We formalize this below.

\begin{theorem}[{\cite[Lemma 1.1]{GanianRamanujanSzeider16}}]
\label{thm:scatcsp}
Let $\Gamma_1,\dots,\Gamma_\ell$ be globally tractable conservative crisp languages over a domain $D$, with each language having arity at most $q$ and containing at most $p$ relations. There exists a function $f$ and an algorithm solving {\VCSP} in time $\bigoh^*(f(\ell,|D|,q,k,p))$, where 
$k$ is the size of a minimum backdoor into $\VCSP[\Gamma_1] \uplus \dots \uplus \VCSP[\Gamma_\ell]$.

\end{theorem}

Observe that in the above theorem, when $q$ and $|D|$ are bounded, $p$ is immediately bounded. However, it is important that we formulate the running time of the algorithm in this form because in the course of our application, these parameters have to be bounded separately.
Our goal for the remainder of this section is to extend Theorem~\ref{thm:scatcsp} in the {\VCSP} setting to also cover infinite globally tractable languages (of bounded arity and domain size). Before proceeding, it will be useful to observe that if each $\Gamma_1,\dots,\Gamma_\ell$ is globally tractable, then the class $\VCSP[\Gamma_1] \uplus \dots \uplus \VCSP[\Gamma_\ell]$ is also tractable (since each connected component can be resolved independently of the others).

\subsection{Finding Backdoors to Scattered Classes}
 \label{sub:findingscat}

In this subsection, we prove that finding backdoors for $\VCSP$ into scattered classes is fixed-parameter tractable. This will then allow us to give a proof of our main theorem, stated below.

   \begin{restatable}{theorem}{finaltheorem}
 \label{thm:vcsp}
Let $\Delta_1,\dots,\Delta_\ell$ be conservative, globally tractable and efficiently recognizable languages over a finite domain and having constant arity. Then $\VCSP$ is fixed-parameter tractable parameterized by the size of a smallest backdoor of the given instance into $\VCSP{(\Delta_1)} \uplus \cdots \uplus \VCSP{(\Delta_\ell)}$.
 \end{restatable}  
 
Recall that the closure of a conservative and globally tractable
language under partial assignments is also a globally tractable
language. Furthermore, every backdoor of the given instance into $\VCSP{(\Delta_1)} \uplus \cdots \uplus \VCSP{(\Delta_\ell)}$ is also a backdoor into $\VCSP{(\Gamma_1)} \uplus \cdots \uplus \VCSP{(\Gamma_\ell)}$ where $\Gamma_i$ is the closure of $\Delta_i$ under partial assignments. Due to Lemma \ref{lem:using}, it follows that it is sufficient to compute a backdoor of small size into the scattered class $\VCSP{(\Gamma_1)} \uplus \cdots \uplus \VCSP{(\Gamma_\ell)}$ where each $\Gamma_i$ is closed under partial assignments.

Our strategy for finding backdoors to scattered classes defined in terms of (potentially infinite) globally tractable languages relies on a two-phase transformation of the input instance. 
In the first phase (Lemma \ref{lem:finitation}), we show that for every choice of $\Gamma_1,\dots,\Gamma_d$ (each having bounded domain size and arity), we can construct a set of finite languages $\Gamma'_1,\dots,\Gamma'_d$ and a new instance $\PP'$ such that there is a one-to-one correspondence between backdoors of~$\PP$ into $\Gamma_1\uplus \dots \uplus \Gamma_d$ and backdoors of~$\PP'$ into $\Gamma'_1\uplus \dots \uplus \Gamma'_d$. This allows us to restrict ourselves to only the case of finite (but not necessarily crisp) languages as far as backdoor detection is concerned. In the second phase (Lemma \ref{lem:vcsptocsp}), we transform the instance and languages one more time to obtain another instance $\PP''$ along with finite crisp languages $\Gamma''_1,\dots,\Gamma''_d$ such that there is a one-to-one correspondence between the backdoors of~$\PP''$ and backdoors of $\PP'$. We crucially note that the newly constructed instances are equivalent \emph{only} with respect to backdoor detection; there is no correspondence between the solutions of these instances.

Before proceeding, we introduce a natural notion of replacement of valued constraints which is used in our proofs.

\begin{definition}\label{def:replace}
  Let $\cP=(V,D,\CC)$ be a $\VCSP$ instance and let
  $c=(\vec{x},\phi)\in \CC$. Let $\phi'$ be a cost function over $D$
  with the same arity as $\phi$. Then the operation of
  \emph{replacing} $\phi$ in $c$ with $\phi'$ results in a new
  instance
  $\cP'=(V,D,({\CC}\setminus \{c\})\cup
  \{(\vec{x},\phi')\})$. 
\end{definition}

 \begin{lemma}
 \label{lem:finitation}
Let $\Gamma_1,\dots,\Gamma_\ell$ be  efficiently recognizable languages closed under partial assignments, each of arity at most $q$ over a domain $D$ of size $d$. There exists an algorithm which takes as input a $\VCSP$ instance $\PP=(V,D,\cal C)$ and an integer $k$, runs in time $\bigoh^*(f(\ell,d,k,q))$ for some function $f$ and either correctly concludes that $\cP$ has no backdoor into $\VCSP(\Gamma_1)\uplus \dots \uplus \VCSP(\Gamma_\ell)$ of size at most $k$ or outputs a $\VCSP$ instance $\PP'=(V,D',\cal C')$ and languages  $\Gamma'_1,\dots,\Gamma'_{\ell}$ with the following properties.
\begin{enumerate}
\item For each $i\in [\ell]$, the arity of $\Gamma_i'$ is at most $q$
\item For each $i\in [\ell]$, $\Gamma'_i$ is over $D'$ and $D'\subseteq D$
\item Each of the languages $\Gamma_1',\dots, \Gamma_\ell'$ is closed under partial assignments and contains at most $g(\ell,d,k,q)$ cost functions for some function $g$.
\item For each $X\subseteq V$, $X$ is a minimal backdoor of $\PP$ into $\VCSP(\Gamma_1)\uplus \dots \uplus \VCSP(\Gamma_\ell)$ of size at most $k$ if and only if $X$ is a minimal backdoor of $\PP'$ into $\VCSP(\Gamma'_1)\uplus \dots \uplus \VCSP(\Gamma'_{\ell})$ of size at most $k$. 
\end{enumerate}
 \end{lemma}
 
 \begin{proof} We will first define a function mapping the valued constraints in $\CC$ to a finite set whose size depends only on $\ell,d,k$ and $q$. Subsequently, we will show that every pair of constraints in $\CC$ which are mapped to the same element of this set are, for our purposes (locating a backdoor), interchangeable. We will then use this observation to define the new instance $\cP'$ and the languages $\Gamma'_1,\dots,\Gamma'_{\ell'}$. To begin with, observe that if the arity of a valued constraint in $\cP$ is at least $q+k+1$, then $\cP$ has no backdoor of size at most $k$ into $\VCSP(\Gamma_1)\uplus \dots \uplus \VCSP(\Gamma_\ell)$. Hence, we may assume without loss of generality that the arity of every valued constraint in $\cP$ is at most $q+k$.
 
Let $\cF$ be the set of all functions from $[q+k]\times
2^{[q+k]}\times D^{[q+k]} \to 2^{[\ell]}\cup \{\bot\}$, where $\bot$ is used a special symbol expressing that $\cF$ is ``out of bounds.''
Observe that $|\cF|\leq \eta(\ell,d,k,q)=(2^\ell+1)^{(2d)^{(q+k)+\log (q+k)}}$. We will now define a function $\og:\CC\to \cF$ as follows. We assume without loss of generality that the variables in the scope of each constraint in $\CC$ are  numbered  from $1$ to $|\var(c)|$ based on their occurrence in the tuple $\vec{x}$ where $c=(\vec{x},\phi)$. Furthermore, recall that  $|\var(c)|\leq q+k$. 
 For $c\in \CC$, we define $\og(c)=\delta\in \cF$ where $\delta$ is defined as follows. Let $r\leq q+k$, $Q\subseteq [q+k]$ and $\gamma:[q+k]\to D$. Let $\gamma[Q\cap [r]]$ denote the restriction of $\gamma$ to the set $Q\cap [r]$. Furthermore, recall that $c|_{\gamma[Q\cap [r]]}$ denotes the valued constraint resulting from applying the partial assignment $\gamma$ on the variables of $c$ corresponding to all those indices in $Q\cap [r]$.

  Then, $\delta(r,Q,\gamma)=\bot$ if $r\neq |\var(c)|$. Otherwise,  $\delta(r,Q,\gamma)=L\subseteq [\ell]$ where $i\in [\ell]$ is in $L$ if and only if $c|_{\gamma[Q\cap [r]]}\in \VCSP({\Gamma_i})$. This completes the description of the function $\og$; observe that $\og(c)$ can be computed in time which is upper-bounded by a function of $\ell, d, k, q$.
 
 For every $\delta\in \cF$, if there is a valued constraint $c\in \CC$ such that $\og(c)=\delta$, we pick and fix one arbitrary such valued constraint $c^*_\delta=(\vec{x^*_\delta},\phi^*_\delta)$. We now proceed to the definition of the instance $\cP'$ and the languages $\Gamma'_1,\dots, \Gamma'_{\ell'}$.
 
 Observe that for 2 constraints $c=(\vec{x_1},\phi),c'=(\vec{x'},\phi')\in \CC$, if $\og(c)=\og(c')$ then $|\var(c)|=|\var(c')|$. Hence, the notion of \emph{replacing} $\phi$ in $c$ with $\phi'$ is well-defined (see Definition \ref{def:replace}). We define the instance $\cP'$ as the instance obtained from $\cP$ by replacing each $c=(\vec{x},\phi)\in \CC$ with the constraint $(\vec{x},\phi^*_\delta)$ where $\delta=\og(c)$.
 
For each $i\in [\ell]$ and cost function $\phi\in  \Gamma_i$,  we add $\phi$ to the language $\Gamma_i'$ if and only if  for some $\delta\in \cF$ and some set $Q\subseteq \var(c^*_\delta)$ and assignment $\gamma:Q\to D$, the constraint $c|_{\gamma[Q]}=(\vec{x}\setminus Q,\phi)$.  
 Clearly, for every $i\in [\ell]$, $|\Gamma_i'|\leq d^q\cdot |\cF|\leq d^q\cdot \eta(\ell,d,k,q)$. Finally, for each $\Gamma_i'$, we compute the closure of $\Gamma_i'$ under partial assignments and add each relation from this closure into $\Gamma_i'$. Since the size of each $\Gamma_i'$ is bounded initially in terms of $\ell,d,k,q$, computing this closure can be done in time $\bigoh^*(\lambda(\ell,d,k,q))$ for some function $\lambda$.
Since each cost function has arity $q$ and domain $D$, the size of the final language $\Gamma_i'$ obtained after this operation is blown up by a factor of at most $d^q$, implying that in the end, $|\Gamma_i'|\leq d^{2q}\cdot |\cF|\leq d^{2q}\cdot \eta(\ell,d,k,q)$.   

Now, observe that the first two statements of the lemma follow from the definition of the languages $\{\Gamma_i'\}_{i\in [\ell]}$. Furthermore, the number of cost functions in each $\Gamma_i'$ is bounded by $d^q\cdot \eta(\ell,d,k,q)$, and so the third statement holds as well.
  Therefore, it only remains to prove the final statement of the lemma. Before we do so, we state a straightforward consequence of the definition of~$\cP'$.
  
  \begin{observation}\label{obs:bijection}
  	For every $Y\subseteq V$, $\gamma:Y\to D$ and connected component $\cH'$ of $\cP'|_\gamma$, there is a connected component $\cH$ of $\cP|_\gamma$ and a bijection $\psi:\cH\to\cH'$ such that for every $c\in \cH$, $\og(c)=\og(\psi(c))$. Furthermore, for every $c=(\vec{x},\phi)\in \cH$, the constraint $\psi(c)$ is obtained by replacing $\phi$ in $c$ with $\phi^*_{\og(c)}$.
  \end{observation}

We now return to the proof of Lemma~\ref{lem:finitation}.
Consider the forward direction and let $X$ be a  backdoor of size at most $k$ for $\cP$ into $\VCSP(\Gamma_1)\uplus \dots \uplus \VCSP(\Gamma_\ell)$ and suppose that $X$ is \emph{not} a backdoor for $\cP'$ into $\VCSP(\Gamma'_1)\uplus \dots \uplus \VCSP(\Gamma_\ell')$. Then, there is an assignment $\gamma:X\to D$ such that for some connected component $\cH$  of $\cP'|_\gamma$, there is no $i\in  \ell$ such that all constraints in $\cH'$ lie in $\VCSP(\Gamma'_i)$. 
By Observation \ref{obs:bijection} above, there is a connected component $\cH$ in $\cP|_\gamma$ and a bijection $\psi:\cH\to\cH'$ such that for every $c\in \cH$, $\og(c)=\og(\psi(c))$.  Since $X$ is a backdoor for $\cP$, there is a $j\in \ell$ such that all constraints in $\cH$ lie in $\VCSP(\Gamma_j)$. Pick an arbitrary constraint $c=(\vec{x},\phi)\in \cH$. Let $c'=(\vec{x},\phi^*_{\og(c)})$ be the constraint $\psi(c)$. By definition of $\phi^*_{\og(c)}$ it follows that $c'|_\gamma\in  \VCSP(\Gamma_j)$. The fact that this holds for an arbitrary constraint in $\cH$ along with the fact that $\psi$ is a bijection implies that every constraint in $\cH'$ is in fact in $\VCSP(\Gamma_j')$, a contradiction. The argument in the converse direction is symmetric. This completes the proof of the final statement of the lemma.

The time taken to compute $\cP'$ and the languages $\Gamma_1',\dots, \Gamma_\ell'$ is dominated by the time required to compute the function $\og$. Since the languages $\Gamma_1,\dots, \Gamma_\ell$ are efficiently recognizable,  this time is bounded by $\bigoh^*(|\cF|)$, completing the proof of the lemma.
 \end{proof}

 \begin{lemma}
 \label{lem:vcsptocsp}
 Let $\Gamma_1,\dots,\Gamma_\ell$ be  efficiently recognizable languages closed under partial assignments, each of arity at most $q$ over a domain $D$ of size $d$. 
 Let $\cP'=(V,D',\CC')$ be the $\VCSP$ instance and let $\Gamma_1',\dots,\Gamma_\ell'$ be languages returned by the algorithm of Lemma \ref{lem:finitation} on input $\cP$ and $k$. There exists an algorithm which takes as input $\cP'$, these languages and $k$, runs in time $\bigoh^*(f(\ell,d,k,q))$ for some function $f$ and outputs a \CSP{} instance $\PP''=(V''\supseteq V,D'',\cal C'')$ and crisp languages $\Gamma_1'',\dots, \Gamma_\ell''$ with the following properties.
 \begin{enumerate}
 \item  For each $i\in [\ell]$, the arity of $\Gamma_i''$ is at most $q+1$
 \item $D''\supseteq D$ and $|D''|\leq \beta(q,d,k)$ for some function $\beta$.
 \item  The number of relations in each of the languages $\Gamma_1'',\dots, \Gamma_\ell''$ is at most $\alpha(q,d,k)$ for some function $\alpha$.
 	\item if $X$ is a minimal backdoor of arity at most $k$ of $\PP''$ into $\CSP(\Gamma''_1)\uplus \dots \uplus \CSP(\Gamma''_{\ell})$, then $X\subseteq V$. 
 	 	\item For each $X\subseteq V$, $X$ is a minimal backdoor of $\PP'$ into $\VCSP(\Gamma_1')\uplus \dots \uplus \VCSP(\Gamma_\ell')$ if and only if $X$ is a minimal backdoor of $\PP''$ into $\CSP(\Gamma''_1)\uplus \dots \uplus \CSP(\Gamma''_{\ell})$.
 \end{enumerate}  
 \end{lemma} 
 
\begin{proof}
We propose a fixed-parameter algorithm $\AA$, and show that it has the claimed properties. It will be useful to recall that we do not distinguish between crisp cost functions and relations. We also formally assume that $D'$ does not intersect the set of rationals $\cal Q$; if this is not the case, then we simply rename elements of $D'$ to make sure that this holds. Within the proof, we will use $\vec{a}\circ b$ to denote the concatenation of vector $\vec{a}$ by element $b$.
 
First, let $T_i$ be the set of all values which are returned by at least one cost function from $\Gamma'_i$, $i\in [\ell]$, for at least one input. Let $T=\bigcup_{i\in [\ell]} T_i$. Observe that $|T|$ is upper-bounded by the size, domain and arity of our languages. Let us now set $D''=D'\cup T\cup \epsilon$. Intuitively, our goal will be to represent the cost function in each valued constraint in $\cP'$ by a crisp cost function with one additional variable which ranges over $T$, where $T$ corresponds to a specific value which occurs in one of our base languages. Note that this satisfies Condition~$2.$ of the lemma, and that $T$ can be computed in linear time from the cost tables of $\Gamma'_1,\dots,\Gamma'_\ell$. We will later construct $k+1$ such representations (each with its own additional variable) to ensure that the additional variables are never selected by minimal backdoors.
 
Next, for each language $\Gamma_i'$, $i\in [\ell]$, we compute a new crisp language $\Gamma_i''$ as follows. For each $\phi\in \Gamma_i'$ of arity $t$, we add a new relation $\psi$ of arity $t+1$ into $\Gamma_i''$, and for each tuple $(x_1,\dots,x_t)$ of elements from $D'$ we add the tuple $(x_1,\dots,x_t,\phi(x_1,\dots,x_t))$ into $\psi$; observe that this relation exactly corresponds to the cost table of $\phi$. 
We then compute the closure of $\Gamma_i''$ under partial assignments and add each relation from this closure into $\Gamma_i''$.
Observe that the number of relations in $\Gamma_i''$ is bounded by a function of $|T|$ and $|\Gamma_i'|$, and furthermore the number of tuples in each relation is upper-bounded by $q^{|D'|}$, and so Conditions~$1.$ and~$3.$ of the lemma hold. The construction of each $\Gamma_i''$ from $\Gamma_i'$ can also be done in linear time from the cost tables of $\Gamma'_1,\dots,\Gamma'_\ell$.

Finally, we construct a new instance $\PP''=(V'',D'',\cal C'')$ from $\PP'=(V,D',\cal C')$ as follows. At the beginning, we set $V'':=V$. For each $c'=(\vec{x}',\phi')\in \CC'$, we add $k+1$ unique new variables $v^1_{c'},\dots,v^{k+1}_{c'}$ into $V''$ and add $k+1$ constraints $c''^{1},\dots,c''^{k+1}$ into $\CC''$. For $i\in [k+1]$, each $c''^{i}=(\vec{x}' \circ v^i_{c'},\psi'')$ where $\psi''$ is a relation that is constructed similarly as the relations in our new languages $\Gamma_i''$ above. Specifically, for each tuple $(x_1,\dots,x_t)$ of elements from $D'$ we add the tuple $(x_1,\dots,x_t,\phi'(x_1,\dots,x_t))$ into $\psi''$, modulo the following exception. If $\phi'(x_1,\dots,x_t)\not \in D''$, then we instead add the tuple $(x_1,\dots,x_t,\epsilon)$ into $\psi''$.
Clearly, the construction of our new instance $\PP''$ takes time at most $\bigoh(|{\cal C'}|+q^{|D'|})$. This concludes the description of $\AA$.

It remains to argue that Conditions~$4.$ and~$5.$ of the lemma hold. 
First, consider a minimal backdoor $X$ of size at most $k$ of $\PP''$ into $\CSP[\Gamma''_1]\uplus \dots \uplus \CSP[\Gamma''_{\ell}]$, and assume for a contradiction that there exists some $c'=(\vec{x}',\phi')\in \CC'$ and $i\in [k+1]$ such that $v^i_{c'}\in X$. First, observe that this cannot happen if the whole scope of $c''^{i}$ lies in $X$.
By the size bound on $X$, there exists $j\in [k+1]$ such that $v^j_{c'}\not \in X$. 
Then for each partial assignment $\tau$ of $X$, the relation $\phi''$ in $c''^j$ belongs to the same globally tractable language as the rest of the connected component of $\cP''$ containing the scope of $c''$ (after applying $\tau$). Since the relation $\phi''$ in $c''^j$ is precisely the same as in $c''^i$ and the scope of $c''^i$ must lie in the same connected component as that of $c''^j$, it follows that $X\setminus \{v_{x'}^i\}$ is also a backdoor of $\PP''$ into $\CSP(\Gamma''_1)\uplus \dots \uplus \CSP(\Gamma''_{\ell})$. However, this contradicts the minimality of $X$.

Finally, for Condition~$5.$, consider an arbitrary backdoor $X$ of $\PP'$ into $\VCSP(\Gamma'_1)\uplus \dots \uplus \VCSP(\Gamma'_{\ell})$, and let us consider an arbitrary assignment from $X$ to $D''$. It will be useful to note that while the contents of relations and/or cost functions in individual (valued) constraints depend on the particular choice of the assignment to $X$, which variables actually occur in individual components depends only on the choice of $X$ and remains the same for arbitrary assignments. 

Now observe that each connected component $\PP^{\text{CSP}}$ of $\PP''$ after the application of the (arbitrarily chosen) assignment will fall into one of the following two cases. $\PP^{\text{CSP}}$ could contain a single variable $v_{c'}$ with a single constraint whose relation lies in every language $\Gamma''_i$, $i\in [\ell]$; this occurs precisely when the whole scope of a valued constraint $c'\in \CC'$ lies in $X$, and the relation will either contain a singleton element from $T$ or be the empty relation. In this case, we immediately conclude that $\PP^{\text{CSP}}\in \CSP(\Gamma_i)$ for each $i\in [\ell]$.

Alternatively, $\PP^{\text{CSP}}$ contains at least one variable $v\in V$. Let $\PP^{\text{VCSP}}$ be the unique connected component of $\PP'$ obtained after the application of an arbitrary assignment from $X$ to $D'$ which contains $v$. Observe that the variable sets of $\PP^{\text{CSP}}$ and $\PP^{\text{VCSP}}$ only differ in the fact that $\PP^{\text{CSP}}$ may contain some of the newly added variables $v_{c'}$ for various constraints $c'$. Now let us consider a concrete assignment $\tau$ from $X$ to $D'$ along with an $i\in [\ell]$ such that after the application of $\tau$, the resulting connected component $\PP^{\text{VCSP}}$ belongs to $\VCSP(\Gamma'_i)$. It follows by our construction that applying the same assignment $\tau$ in $\PP''$ will result in a connected component $\PP^{\text{CSP}}$ corresponding to $\PP^{\text{VCSP}}$ such that $\PP^{\text{VCSP}}\in \CSP(\Gamma''_i)$; indeed, whenever $\Gamma'_i$ contains an arbitrary cost function $\phi(\vec{x})=\beta$, the language $\Gamma''_i$ will contain the relation $(\vec{x}\circ\beta)$. 

By the above, the application of an assignment from $X$ to $D'$ in $\PP''$ will indeed result in an instance in $\CSP(\Gamma_1''\uplus \dots \uplus \Gamma_\ell'')$. But recall that the domain of $\PP''$ is $D''$, which is a superset of $D'$; we need to argue that the above also holds for assignments $\tau$ from $X$ to $D''$. To this end, consider an arbitrary such $\tau$ and let $\tau_0$ be an arbitrary assignment from $X$ to $D'$ which matches $\tau$ on all mappings into $D'$. Let us compare the instances $\PP''_{\tau_0}$ and $\PP''_{\tau}$. By our construction of $\PP'$, whenever $\tau$ maps at least one variable from the scope of some constraint $c''$ to $D''\setminus D'$, the resulting relation will be the empty relation. It follows that each constraint in $\PP''_{\tau}$ will either be the same as in $\PP''_{\tau}$, or will contain the empty relation. But since the empty relation is included in every language $\Gamma_1'',\dots,\Gamma_\ell''$, we conclude that each connected component of $\PP''_\tau$ must belong to at least one language $\Gamma_i''$, $i\in[\ell]$. This shows that $X$ must also be a backdoor of $\PP''$ into $\CSP[\Gamma''_1]\uplus \dots \uplus \CSP[\Gamma''_\ell]$.

For the converse direction, consider a minimal backdoor $X$ of $\PP''$ into $\CSP[\Gamma''_1]\uplus \dots \uplus \CSP[\Gamma''_\ell]$. Since we already know that Condition~$4.$ holds, $X$ must be a subset of~$V$. The argument from the previous case can then simply be reversed to see that $X$ will also be a backdoor of $\PP'$ into $\VCSP[\Gamma'_1]\uplus \dots \uplus \VCSP[\Gamma'_\ell]$; in fact, the situation in this case is much easier since only assignments into $D'$ need to be considered.

Summarizing, we gave a fixed-parameter algorithm and then showed that it satisfies each of the required conditions, and so the proof is complete.   
\end{proof}
We are now ready to prove Theorem \ref{thm:vcsp}, which we restate for the sake of convenience.
 
\finaltheorem*

\begin{proof}
 For each $i\in[\ell]$, let $\Gamma_i$ denote the closure of $\Delta_i$ under partial assignments. Observe that every backdoor of the given instance into $\VCSP{(\Delta_1)} \uplus \cdots \uplus \VCSP{(\Delta_\ell)}$ is also a backdoor into $\VCSP{(\Gamma_1)} \uplus \cdots \uplus \VCSP{(\Gamma_\ell)}$. Furthermore, each $\VCSP{(\Gamma_i)}$ is tractable since $\VCSP{(\Delta_1)}$ is tractable and conservative. Hence, it is sufficient to compute and use a backdoor of size at most $k$ into $\VCSP{(\Gamma_1)} \uplus \cdots \uplus \VCSP{(\Gamma_\ell)}$.
 
  The claimed algorithm has two parts.
The first one is finding a backdoor into $\VCSP{(\Gamma_1)} \uplus \cdots \uplus \VCSP{(\Gamma_\ell)}$ and the second one is using the computed backdoor to solve $\VCSP$.
 Given an instance $\cP$ and $k$, we first execute the algorithm of Lemma \ref{lem:finitation} to compute the instance $\cP'$, and the languages $\Gamma_1',\dots, \Gamma_\ell'$ with the properties stated in the lemma. We then execute the algorithm of Lemma \ref{lem:vcsptocsp} with input $\cP'$, $k$, and $\Gamma_1',\dots, \Gamma_\ell'$ to compute the CSP instance $\cP''$ and crisp languages $\Gamma_1'',\dots, \Gamma_\ell''$ with the stated properties. Following this, we execute the algorithm of Theorem \ref{thm:scatcsp} with input $\cP'',k$. If this algorithm returns {\No} then we return {\No} as well. Otherwise we return the set returned by this algorithm as a backdoor of size at most $k$ for the given instance $\cP$. Finally, we use the algorithm of Lemma \ref{lem:using} with $\cal H$ set to be the class $\VCSP{(\Gamma_1)} \uplus \cdots \uplus \VCSP{(\Gamma_\ell)}$,  to solve the given instance.

 The correctness as well as running time bounds follow from those of
 Lemmas~\ref{lem:finitation} and~\ref{lem:vcsptocsp}, Theorem~\ref{thm:scatcsp}, and Lemma~\ref{lem:using}. This completes the proof of the theorem.
\end{proof}

 \section{Concluding Remarks}
 \label{sec:conclusion}
We have introduced the notion of backdoors to the  VCSP setting as a means for
  augmenting a class of globally tractable VCSP instances to instances that are
  outside the class but of small distance to the class. We have
  presented fixed-parameter tractability results for solving VCSP
  instances parameterized by the size of a smallest backdoor into a
  (possibly scattered and heterogeneous) tractable class satisfying certain natural properties.

   Our work opens up several avenues for future research. 
Since our main objective was to establish the fixed-parameter tractability of this problem, we have not attempted to optimize the runtime bounds for finding backdoors to scattered classes.
As a result, it is quite likely that a more focussed study of scattered classes arising from specific constraint languages will yield a significantly better runtime. A second interesting direction would be studying the parameterized complexity of detection of backdoors into tractable VCSP classes that are
characterized by specific fractional polymorphisms.

\bibliographystyle{abbrv} \bibliography{literature}

\end{document}